\documentclass[11pt]{scrartcl}

\usepackage{amsmath}
\usepackage[english]{babel}
\usepackage{graphicx, dsfont, amssymb, amsthm, enumerate, units, breakcites}

\newcommand{\R}{\ensuremath{{\mathbb R}}}
\newcommand{\E}{\ensuremath{{\mathbb E}}}
\newcommand{\N}{\ensuremath{{\mathbb N}}}

\newtheorem{theorem}{Theorem}[section]
\newtheorem{definition}{Definition} [section]
\newtheorem{lemma}[theorem]{Lemma}
\newtheorem{proposition}[theorem]{Proposition}

\newtheorem{remark}[theorem]{Remark}

\begin{document}

\begin{center}
{\Large{\textbf{Optimal relaxed portfolio strategies for growth rate maximization problems with transaction costs}}}
\end{center}
\begin{center}\vspace{.5cm}{\large{ S\"oren Christensen\textsuperscript{ *}}} \\
\large\textit{Christian-Albrechts-University Kiel\let\thefootnote\relax\footnotetext{\textsuperscript{*} Mathematisches Seminar, Christian-Albrechts-Universit\"at zu Kiel, Ludewig-Meyn-Str. 4,
24098 Kiel, Germany, e-mail: christensen@math.uni-kiel.de }\\}
\vspace{.8cm}{\large{ Marc Wittlinger\textsuperscript{ **}}}\\
\large\textit{Ulm University \let\thefootnote\relax\footnotetext{\textsuperscript{**} Institute of Mathematical Finance, Ulm University, Helmholtzstra\ss e 18, 89081 Ulm, Germany, email: marc.wittlinger@uni-ulm.de }\\}

\end{center}

\vspace*{4mm}

In this paper we investigate a new class of growth rate maximization problems based on impulse control strategies such that the average number of trades per time unit does not exceed a fixed level. Moreover, we include proportional transaction costs to make the portfolio problem more realistic. We provide a {Verification Theorem}  to compute the optimal growth rate as well as an optimal trading strategy. Furthermore, we prove the existence of a constant boundary strategy which is optimal. At the end, we compare our approach to other discrete-time growth rate maximization problems in numerical examples. It turns out that constant boundary strategies with a small average number of trades per unit perform nearly as good as the classical optimal solutions with infinite activity. \vspace{.8cm}

{ Key Words:} portfolio optimization; growth rate; optimal time discretization; impulse control; constant boundary strategies; {transaction costs}\\[2mm]

{ AMS subject classification:}   49N25, 93E20, 91G10

\section{Introduction}
One of the fundamental results in portfolio optimization goes back to {\cite{RMP}}. He proved in the framework of a Black Scholes market that the expected utility can be maximized by keeping the risky fraction process, which indicates  the fraction of the investor's wealth invested in the stock, constant. The same result holds true if the objective is changed to the asymptotic growth rate, i.e. the investor maximizes
\[\liminf_{T\rightarrow\infty}\frac{1}{T}\;\E \big{[} \log (X_T)\big{]},\]
where $X_T$ denotes the investor's wealth at time $T$. This criterion goes back to Kelly (\cite{K}) and will be our objective in the following considerations. The optimal strategy of keeping the risky fraction process constant has one large drawback. The investor has to trade at all times since the stock price fluctuates at all times and so the investor's wealth changes at all times, too. Obviously this is not realizable in practice.

A first approach to avoid this problem is to restrict trading to discrete time points. This approach may be accomplished in different ways. For deterministic equidistant time points, this was carried out in \cite{relaxed}. A further idea of limiting the trading times is the widely used approach to restrict the trading times to those at which an independent Poisson process jumps. Such situations are for example studied in \cite{MatLog, MatGrow} and \cite{PTCM}. Both {approaches }are revised in Section \ref{sec:discrete}.

A second approach is based on the fact that transaction costs play a major role in real financial markets. The most natural type of costs is given by proportional transaction costs, that is the investor has to pay a fixed fraction of his transaction volume at each transaction. There is a wealth of papers dealing with different aspects of models with proportional transaction costs. For example see \cite{assaf}, \cite{DavisNorman}, \cite{SS}, {\cite{AST},  \cite{KK}, \cite{KM}.} By introducing the mentioned proportional transaction costs into the model, transactions of infinite variation lead to instantaneous ruin, and therefore cannot be optimal anymore. Moreover, it turns out that the optimal strategy is given in the following way: If the risky fraction process evolves within a predetermined interval, then it is optimal not to trade. If the risky fraction process reaches one of the interval boundaries, then infinitesimal trading occurs such that the fraction of wealth is reflected at this boundary. Unfortunately, this strategy is still not realizable in practice, since it involves continuous trading at the boundaries.{ One way
out is to introduce artificial fixed transaction costs, which punish high frequent trading, see {\cite{MP}, \cite{K98}, \cite{OS}, }\cite{IS06,IS06_2}, \cite{T06,T08}, and \cite{DDS}}. The advantage of using fixed transaction costs that are proportional to the investors wealth is that the optimal strategies turn out to be realizable and easy to describe: They are determined by four parameters $a<\alpha\leq \beta<b$. Whenever the risky fraction process reaches one of the boundaries $a$ resp. $b$, the investor trades such that the new risky fraction equals $\alpha$ resp. $\beta$. These strategies are called constant boundary strategies and fall into the class of impulse control strategies. A drawback of this model is that the use of fixed transaction cost cannot be justified in most real world situations.

The aim of this paper is bringing together the two approaches described above. We consider a portfolio optimization model with the realistic pure proportional transaction costs (including the case of no transaction costs). Then we restrict the set of trading strategies to {relaxed strategies such} that the average number of trades per time unit does not exceed a fixed level $\nicefrac{1}{h}$, that is
\begin{equation}\label{eq:strategies}
\limsup_{T\rightarrow\infty} \;\frac{1}{T}\;\E(|\{n:\tau_n\leq T\}|)\;\le\; \frac{1}{h},
\end{equation}
where the (random) trading points $\tau_0,\tau_1,...$ are modeled as stopping times and $h$ is some positive constant. {The advantages of considering this class of strategies is the fact that they are easy to implement and they do not include unrealizable high frequency trading.} {On the other hand, they are very flexible.} For example, whenever the asset price fluctuates a lot, the investor adjusts his portfolio necessarily very often; and in flat market times, the investor can be relaxed. From a mathematical point of view, the problem we are faced with is an impulse control problem where the policy set is restricted according to \eqref{eq:strategies}. To the best of the authors' knowledge, such problems have not been discussed in the literature, yet. We provide a Verification Theorem to verify the growth rate as well as an optimal strategy. Moreover, we prove the existence of an optimal strategy that is of constant boundary type, whose parameters can be computed as a solution to a certain system of equations. It turns out that this system of equations is similar to those considered for problems with both fixed and proportional transaction costs with one additional free parameter.

The structure of this article is as follows. In the short Section \ref{sec:merton}, we introduce the Merton problem and state its well known  solution with and without transaction costs for later reference. Section \ref{optimal_investor} contains the main results of this article. After formulating the problem and collecting some known facts about constant boundary strategies, we formulate the Verification Theorem and apply it to our setting by showing that constant boundary strategies are optimal. This is remarkable, since these strategies are very easy to describe and may be considered as easier to handle than the discrete-time strategies discussed in Section \ref{sec:discrete}. We end this section by treating an explicit example for a certain set of parameters. In the final Section \ref{sec:example}, we compare the asymptotic growth rates for the optimal strategy obtained in Section \ref{optimal_investor} and the existing results from Subsection \ref{sec:discrete}.

It turns out, that the optimal strategy performs much better than the other strategies for different realistic sets of parameters. More important, one can see that in a market with proportional transaction costs, the \glqq relaxed\grqq\; investor who uses a constant boundary strategy nearly performs as good as an investor who uses the optimal strategy without restrictions described above, even for larger values of $h$.

\section{The Merton problem}\label{sec:merton}
Let us consider a market in which two assets are traded continuously. One of the assets is a risky asset, called stock, with price
$ S=(S_t)_{t\ge 0} $ satisfying the stochastic differential equation
\begin{equation*}
dS_t = \mu S_tdt + \sigma S_tdW_t\;, \qquad   S_{0} = s_0 > 0\;,
\end{equation*}
where $ \mu $ and $ \sigma $ are constants and $ W=(W_t)_{t\ge 0} $ is a one dimensional standard Brownian motion on a complete filtered probability space $ (\Omega,\mathcal{F},(\mathcal{F}_t)_{t\ge 0},\mathds{P}) $ satisfying the usual conditions. The other asset, called bond, is a riskless asset which is chosen as numeraire. Hence the bond has the price $B=(B_t)_{t\ge 0} \equiv  1$. We imagine an investor with initial wealth $ x_0 >0 $ who invests his wealth in the above market and call him a {\bf Merton-investor}.

\subsection{Merton-investor without transaction costs}
Since we assume no transaction costs and a self-financing portfolio, the Merton-investor's wealth process $X=(X_t)_{t\ge 0}$ is given by
\begin{equation*}
X^{\pi}_{t} = x_0 \cdot \exp \bigg{[} \int \limits_0^t g(\pi_s) ds + \int \limits_0^t \pi_s \sigma dW_s  \bigg{]}\;,
\end{equation*}
where the trading strategy $\pi = (\pi_t)_{t\ge 0}$ is assumed to be a predicable process which indicates the fraction of wealth invested in the stock at time $t$ and $g(x)=x(\mu-\frac{1}{2}\sigma^2x)$. Finally, the Merton-investor is going to maximize his growth rate, i.e. he is interested in
\begin{align}
V^{M} = \sup_{\pi}\liminf_{T\rightarrow \infty} \frac{1}{T}\E\big{(} \log(X^{\pi}_T) \big{)} \;.
\label{merton}
\end{align}
By a pointwise maximization, it directly follows that the solution of (\ref{merton}) is given by
\begin{align*}
V^{M} = \liminf_{T\rightarrow \infty} \frac{1}{T}\E\big{(} \log(X^{\pi^{*}}_T) \big{)}= \frac{\mu^2}{2 \sigma^2}\;,
\end{align*}
where the optimal trading strategy $\pi^{*} = (\pi_t)_{t\ge 0}$ equals the Merton ratio $\tfrac{\mu}{\sigma^2}$.

In the special case when $\mu=8\%$ and $\sigma=40\%$, the optimal growth rate equals $2\%$. Further the investor has to put at all times $50\%$ of his wealth in the stock.

\subsection{Merton-investor with proportional transaction costs}
As in \cite{AST} we consider the Merton problem with proportional transaction costs. This means that we have to pay a fixed fraction $\gamma\in[0,1)$ for each transaction volume from the bond account. The corresponding asset price dynamics are now given by the following SDEs:
\begin{align*}
dB_t &= (1-\gamma)dU_t -(1+\gamma)dZ_t\;,  \qquad \qquad \;  B_{0}= 1\;, \\
dS_t &= \mu S_tdt + \sigma S_tdW_t+dZ_t - dU_t\;,   \qquad  S_{0} = s_0 > 0\;,
\end{align*}
where $Z,U$ are non-decreasing adapted c\`{a}dl\`{a}g processes representing the cumulative purchase and sale of the stock at time $t$, respectively. Under the assumption that there is no borrowing and shortselling, the Merton-investor is going to maximize his growth rate, i.e. he is interested in
\begin{align}
V^{M}_c = \sup_{(Z,U)}\liminf_{T\rightarrow \infty} \frac{1}{T}\E\big{(} \log(X^{\pi}_T) \big{)} \;.
\label{merton2}
\end{align}

From \cite{AST}[Section 9.2] we get an explicit solution of problem (\ref{merton2}) when the Merton ratio $\pi^{*}=\tfrac{\mu}{\sigma^2}\in(0,1)$. Let $f:\mathds{R}\rightarrow\mathds{R}$ be given by \begin{align*}
f(x) = \bigg{(}\frac{2\pi^{*}-x}{x} \bigg{)}^{\frac{2\pi^{*}}{2\pi^{*}-1}}\cdot\bigg{(}\frac{1-2\pi^{*}+x}{1-x} \bigg{)}^{\frac{2(1-\pi^{*})}{2\pi^{*}-1}}-1-\frac{2\gamma}{1-\gamma}\;.
\end{align*}
Further let $b$ be the zero of $f$ on $[0,1]$ and $a:=[2\pi^{*}-b]/[1+\tfrac{2\gamma}{1-\gamma}(1-2\pi^{*}+b)]$. Then the optimal growth rate $V^{M}_c$ equals $b\sigma^2(\pi^{*}-b/2)$ and the optimal policy is a constant boundary policy with upper bound $b$ and lower bound $a$, i.e.
\begin{itemize}
\item[-] If the fraction of wealth invested in the stock lies in $(a,b)$ then the investor does not trade. $(a,b)$ is called no transaction region.
\item[-] If the fraction process equals one of the bounds $a$ or $b$ then the investor starts continuous-time trading such the the fraction process does not exit the no transaction region.
\item[-] If the fraction of wealth invested in the stock lies outside the interval $[a,b]$ then the investor instantaneously trades to bring the fraction process to the closest boundary $a$ or $b$.
\end{itemize}

\section{The optimal relaxed investor with transaction costs}\label{optimal_investor}
Now, we discuss the question whether there are optimal strategies in the class of strategies where the average number of transactions per time unit is limited to $1/h$. To answer this question, we first model this situation, and then solve the problem in this section.\\
The class of strategies we consider are impulse control strategies, i.e. sequences $K=(\tau_n,\eta_n)_{n\in\N_0}$ of stopping times $0=\tau_0\leq \tau_1\leq ...\nearrow\infty$ with $\tau_n<\tau_{n+1}$ on $\{\tau_n<\infty\}$ for all $n$, and $\mathcal{F}_{\tau_n}$-measurable random variables $\eta_n\in (0,1)$. The stopping times $\tau_n$ describe the trading times, i.e. the time points when we adjust the fraction of wealth invested in the stock; $\eta_n$ is the new fraction of wealth invested in the stock at $\tau_n$, i.e. $X_{\tau_n}=\eta_n$. Between each two trading times, the risky fraction process runs uncontrolled with dynamic
\[d\pi_t=\pi_t(1-\pi_t)(\mu-\sigma^2\pi_t)dt+\pi_t(1-\pi_t)\sigma dW_t,\;\;\;t\geq 0.\]
We denote the set of all impulse control strategies by $\mathcal{K}$ and denote by $\E^K$ the expectation for the process controlled according to $K=(\tau_n,\eta_n)_{n\in\N_0}$.\\
As discussed in the introduction, we want to introduce possible proportional transaction costs into our model, that is, at each trading time we have to pay a fixed fraction $\gamma\in[0,1)$ of our transaction volume from our bond account. For a more detailed discussion of the transaction cost structure, the impulse control strategies and a comparison to other representation of the strategies, we refer to \cite{IS06} and \cite{IS06_2}, Section 3. In the later reference, the following  result on the asymptotic growth rate is found as an easy application of It\^o's formula -- see \cite{IS06_2}, Proposition 3.1, Formula (3.11):

\begin{proposition}\label{repr_growth}
 For each $K\in\mathcal{K}$ the asymptotic growth rate of the controlled process with starting state $\pi$ is given by
 \[r_K=\liminf_{t\rightarrow\infty}\frac{1}{T}\E^K_\pi\left(\int_0^T g(\pi_t)dt+\sum_{n:\tau_n\leq T}\Gamma(\pi_{\tau_n},\eta_n)\right),\]
 where \[g(x)=x(\mu-\frac{1}{2}\sigma^2x)\;\;\mbox{ and }\;\;
\Gamma(x,y)=\begin{cases}
 \log\frac{1-\gamma x}{1-\gamma y},\;\;\;&y<x,\\
  \log\frac{1+\gamma x}{1+\gamma y},\;\;\;&y\geq x.
 \end{cases}\]
 $r_K$ is independent of the starting state $\pi$.
 \end{proposition}

 Although, impulse control strategies restrict trading times to discrete time points, there is no restriction to the average number of transaction per time unit. Obviously, this is not realistic. On the other hand, no investor would restrict himself to trade only at fixed time point (say, weekly). Under some circumstances, he will adjust his portfolio quite often, under other circumstances, he will perhaps stop trading for a longer time, depending on the market behavior. But in the long time average, we assume that he wants to trade not more than each $h$ time units. Therefore, for  $h>0$ we  introduce the set $\mathcal{K}_h$ of impulse control strategies $K=(\tau_n,\eta_n)_{n\in \N}$ fulfilling
\[\limsup_{T\rightarrow\infty} \frac{1}{T}\E^K_\pi(|\{n:\tau_n\leq T\}|)\le \frac{1}{h}{\mbox{ for all }\pi},\]
that is, in the long term average we trade not more than $1/h$ times per time unit. Our goal is to maximize the asymptotic growth rate in the class $\mathcal{K}_h$, i.e., we want to express
\begin{equation}\label{eq:value}
V^{o}(h):=\sup_{K\in\mathcal{K}_h}r_K
\end{equation}
explicitly, and to find $K\in\mathcal{K}_h$ such that $V^{o}(h)=r_K$.

 \subsection{Constant boundary strategies}\label{subs:constant_boundary}
Natural examples of impulse control strategies are given by constant boundary strategies. These strategies are given by four parameters $0<a<\alpha\leq \beta<b<1$ as follows: Whenever the risky fraction process $(\pi_t)_{t\geq 0}$ falls below $a$, the process is shifted back to $\alpha$, and whenever it exceeds $b$, it is shifted to $\beta$, formally
\begin{align*}
&\tau_n=\inf\{t\geq \tau_{n-1}:\pi_t\not\in(a,b)\}\\
&\eta_n=\phi(\pi_{\tau_n}),
\end{align*}
with $\phi(x)=\alpha$ for $x\leq a$ and $\phi(x)=\beta$ for $x\geq b$. In many impulse control problems, constant boundary strategies turn out to be optimal, see \cite{Ko} for an overview. In our setting (with additional fixed transaction costs) this class of strategies was studied in detail in \cite{IS06} using renewal theoretic arguments. In the following lemma, we collect the most important facts for our considerations.

\begin{proposition}\label{prop:renewal}
Let $K=(\tau_n,\eta_n)_{n\in\N_0}$ be a constant boundary strategy with parameters $(a,\alpha, \beta,b)$.
\begin{enumerate}[(i)]
\item For all $\pi \in [a,b]$, it holds that
\[{\E^{K}_\pi({\tau_1})}=\frac{h_0(\pi)-h_0(a)}{h_0(b)-h_0(a)}(h_1(b)-h_1(a))+h_1(a)-h_1(\pi),\]
where
\[h_0(x)=\begin{cases}
-\left(\frac{1-x}{x}\right)^{2\pi^*-1}\;\;\;&\pi^*\not=\frac{1}{2}\\
\log\frac{1-x}{x}&\pi^*=\frac{1}{2},
\end{cases},\;\;\;\;h_1(x)=\begin{cases}
-\frac{2\log\frac{1-x}{x}}{\sigma^2(2\pi^*-1)}\;\;\;&\pi^*\not=\frac{1}{2}\\
\frac{1}{\sigma^2}\left(\log\frac{1-x}{x}\right)^{2}&\pi^*=\frac{1}{2}
\end{cases}
\]
{and $\pi^{*}$ denotes the Merton ratio.}
\item There exists a probability distribution $\nu$ on $\{\alpha,\beta\}$ (called the invariant distribution of $K$), such that
\[r_K=\frac{\E^{{K}}_\nu q(\pi_{{\tau_1}},\phi(\pi_{{\tau_1}}))}{\E^{K}_\nu {(\tau_1)}},\]
where $\phi$ is the function as before and
\[q(\pi,\eta)=\log\frac{1-\eta}{1-\pi}+\Gamma(\pi,\eta).\]
{More explicitly, $\nu$ is given by
\[\nu(\{\alpha\})=\frac{h_0(b)-h_0(\beta)}{h_0(\alpha)-h_0(a)+h_0(b)-h_0(\beta)},\;\nu(\{\beta\})=1-\nu(\{\alpha\}).\]}
\item \label{renewal:(iii)}For all\ $\pi\in[a,b]$
\[\lim_{T\rightarrow\infty} \frac{1}{T}\E^K_\pi(|\{n:\tau_n\leq T\}|)=\frac{1}{\E^{K}_\nu({\tau_1})}.\]
\item \label{renewal:(iv)}
\[{\E^{{K}}_\nu {(\tau_1)}}={p(h_1(a)-h_1(\alpha))+(1-p)(h_1(b)-h_1(\beta))},\]
where
\[p=\nu(\{\alpha\})=\frac{h_0(b)-h_0(\beta)}{h_0(\alpha)-h_0(a)+h_0(b)-h_0(\beta)}.\]
\end{enumerate}
\end{proposition}

\begin{proof}
All results can be found in \cite{IS06}: (i) is Lemma 3, (ii) is Theorem C.2, (iii) is proved in the last lines of the proof of Theorem C.2, and (iv) holds by ibid., Corollary 1.
\end{proof}

 \subsection{Verification Theorem}
We denote the infinitesimal operator of the uncontrolled process $(\pi_t)_{t\geq0}$ by $L$, i.e.
\[Lv(\pi)=\mu(\pi)v'(\pi)+\frac{1}{2}\sigma^2(\pi)v''(\pi)\;\;\;{\mbox{ for $\pi\in(0,1),v\in C^2$ around $\pi$}},\]
where $\mu(\pi)=\pi(1-\pi)(\mu-\sigma^2\pi)$ and $\sigma(\pi)=\sigma\pi(1-\pi).$
\begin{theorem}\label{thm:verification}
Let $h>0$.
\begin{enumerate}[(a)]
\item \label{ver_a} Let $v:(0,1)\rightarrow\R$, $\lambda\in\R,c\geq 0$ such that
\begin{enumerate}[(i)]
\item $v$ is $C^1$, and piecewise $C^2$,
\item $v(\pi')-v(\pi)+\Gamma(\pi,\pi')-c \leq 0$ for all  $\pi,\pi'\in (0,1)$,
\item $Lv(\pi)+g(\pi)-\lambda\leq 0$ for all $\pi\in(0,1)$,
\item $M:=\left(\int_0^t\sigma(\pi_s)v'(\pi_s)dW_s\right)_{t\geq 0}$ is a martingale.
\end{enumerate}
Then
\[V^{o}(h):=\sup_{K\in \mathcal{K}_h}r_K\leq \lambda+c/h.\]
\item\label{ver_b} If there furthermore exist $0<a=a_h< \alpha=\alpha_h\leq\beta=\beta_h< b=b_h<1$ such that
\begin{enumerate}[(i)]
\setcounter{enumii}{4}
\item $Lv(\pi)+g(\pi)-\lambda=0$ for all $\pi\in(a,b)$,
\item $v(\alpha)-v(a)+\Gamma(a,\alpha)-c= 0=v(\beta)-v(b)+\Gamma(b,\beta)-c$,
\item\label{renewal} For the constant boundary strategy $K=(\tau_n,\eta_n)_{n}$ with parameters $(a,\alpha,\beta,b)$ it holds that
\[\liminf_{T\rightarrow\infty} \frac{1}{T}\E^K(|\{n:\tau_n\leq T\}|)= \frac{1}{h},\]
\end{enumerate}
then
\[V^{o}(h)=\lambda+\frac{c}{h},\]
and the constant boundary strategy with parameters $(a,\alpha,\beta,b)$ is optimal.
\end{enumerate}
\end{theorem}

\begin{proof}
Applying It\^o's formula for jump diffusion processes yields
\[v(\pi_{T+})=v(\pi_0)+\int_0^TLv(\pi_t)dt+M_T+\sum_{t\leq T}(v(\pi_{t+})-v(\pi_t))\]
for each $T>0$, where $M = (M_t)_{t\ge 0}$ is a martingale. {Therefore, }for each ${K=(\tau_n,\eta_n)_{n\in\N_0}}\in\mathcal{K}_h$ and each $T$ we have

\begin{align*}
\int_0^Tg(\pi_t)dt+\sum_{n:\tau_n\leq T}\Gamma(\pi_{\tau_n},\eta_n)=& v(\pi_0)-v(\pi_{T+})+\int_0^T(Lv(\pi_t)+g(\pi_t)-\lambda) dt+M_T\\
& +\sum_{t\leq T}(v(\pi_{t+})-v(\pi_t)+\Gamma(\pi_{t},\pi_{t+}))+\lambda T\\
\leq&\;\; v(\pi_0)-v(\pi_{T+})+M_T\\
& +\sum_{t\leq T}(v(\pi_{t+})-v(\pi_t)+\Gamma(\pi_{t},\pi_{t+}))+\lambda T\\
\leq& v(\pi_0)-v(\pi_{T+})+c|\{n:\tau_n\leq T\}|+M_T+\lambda T.
\end{align*}
Hence,
\[\liminf_{T\rightarrow\infty}\frac{1}{T}\E^{K}_\pi\left(\int_0^Tg(\pi_t)dt+\sum_{n:\tau_n\leq T}\Gamma(\pi_{\tau_n},\eta_n)\right)\leq \lambda+c/h.\]
Using Proposition \ref{repr_growth}, this proves \eqref{ver_a}. For \eqref{ver_b} note that under the stated assumption for the constant boundary strategy $K$ with parameters $(a,\alpha,\beta,b)$ equality holds in each step.
\end{proof}

\begin{remark}
Obviously, the forgoing proof had nothing to do with the special situation we are faced with. Therefore, the previous Verification Theorem is applicable {for a large class of problems with an underlying }It\^o diffusion on an interval.
\end{remark}

\subsection{Existence of a solution}
Now, we find an optimal solution as indicated in Theorem \ref{thm:verification}. In the following, our standing assumption is that the Merton ratio $\pi^*\in(0,1)$. In a first step, we fix a constant $c>0$ and find constants such that conditions $(i)$-$(vi)$ in Theorem \ref{thm:verification} are fulfilled. These constants can be found by standard arguments involving the solution to ODEs. As presented in the following Lemma, we can shorten this {technical} discussion, since we can exactly follow the arguments given in \cite{IS06_2}:
\begin{lemma}\label{lem:fixed_costs}
For a fixed $c>0$, there exists $v_c:(0,1)\rightarrow\R$, $\lambda_c\in\R$, and $0<a_c< \alpha_c\leq\beta_c< b_c<1$ such that assumptions $(i)$-$(vi)$ in Theorem \ref{thm:verification} hold true. The constants $\lambda_c\in\R$, and $0<a_c< \alpha_c\leq\beta_c< b_c<1$ are the unique solution to the following system of equations:
\begin{align*}
k(b)&=-\frac{\gamma}{1-\gamma b}\\
k({a})&=\frac{\gamma}{1+\gamma a}\\
k(\beta)&=-\frac{\gamma}{1-\gamma \beta}\\
k(\alpha)&=\frac{\gamma}{1+\gamma \alpha}\\
\int_\beta^bk(x)dx&=\Gamma(b,\beta)+\log(1-\delta_c)\\
\int_\alpha^ak(x)dx&=\Gamma(a,\alpha)+\log(1-\delta_c)\\
\end{align*}
where {$\delta_c:=\delta:=1-e^{-c}\in(0,1)$, and}
$k(x)=k_c(x)$ is the function given in \cite[Formula (6.16)]{IS06_2}, that continuously depend on c.
\end{lemma}

\begin{proof}
We define the new cost function $\Gamma_c$ by
\[\Gamma_c(x,y)=\Gamma(x,y)-c=\begin{cases}
 \log\frac{1-\gamma x}{1-\gamma y}+\log(1-\delta_c),\;\;\;&y<x,\\
  \log\frac{1+\gamma x}{1+\gamma y}+\log(1-\delta_c),\;\;\;&y\geq x.
 \end{cases},\]
where $\delta_c:=\delta:=1-e^{-c}\in(0,1)$ and consider the impulse control problem given by
\begin{equation}\label{eq:modified}
r_c:=\sup_{K\in\mathcal{K}}\liminf_{t\rightarrow\infty}\frac{1}{T}\E^{K}_\pi\left(\int_0^Tg(\pi_t)dt+\sum_{n:\tau_n\leq T}\Gamma_c(\pi_{\tau_n},\eta_n)\right).
\end{equation}
This is a slight modification of the problem discussed in \cite[Section 5-7]{IS06_2}; the cost structure discussed there was a bit more difficult, to be precise the problem was solved for
\[\Gamma_{IS}(x,y)=
\begin{cases}
 \log\frac{1-\delta-\gamma x}{1-\gamma y},\;\;\;&y<x,\\
  \log\frac{1-\delta+\gamma x}{1+\gamma y},\;\;\;&y\geq x.
 \end{cases}\]
 Nonetheless, the discussion given there immediately applies to this cost structure too, and one obtains the given result by following their line of argument in Section 6-7. This leads to our result.
\end{proof}

\begin{remark}\label{anmerk:tamura}
The problem \eqref{eq:modified} was furthermore studied in \cite{T06} and \cite{T08} (with exactly this cost structure) using quasi-variational inequality techniques. This approach could also be used {to obtain Lemma \ref{lem:fixed_costs} by a slight extension of the arguments given in \cite[Chapter 3]{L}.}
\end{remark}

As a second step, we now consider for each $c$, the optimal impulse control strategy $K_c=(\tau_{c,n},\eta_{c,n})_{n\in\N}$ with parameters $(a_c, \alpha_c,\beta_c, b_c)$ as given in Lemma \ref{lem:fixed_costs}. To fulfill condition (vii) in Theorem \ref{thm:verification}, i.e.
\[\liminf_{T\rightarrow\infty} \frac{1}{T}\E^{K_c}_\pi(|\{n:\tau_n\leq T\}|)= \frac{1}{\mathds{E}^{K_c}_{\nu_c} (\tau_{c,1})}=\frac{1}{h}\]
for a given $h>0$, we want to utilize Proposition \ref{prop:renewal} (iii) and consider {$h =\E^{K_c}_{\nu_c}(\tau_{c,1})$}, where $\nu_c$ denotes the invariant distribution associated to $K_c$ as defined in Proposition \ref{prop:renewal}.
\begin{lemma}\label{lem:h_c}
For  each $h>0$ there exists $c>0$ such that {$h=\E^{K_c}_{\nu_c}(\tau_{c,1})$}.
\end{lemma}

\begin{proof}
First, we show that $\E^{K_c}_{\nu_c}(\tau_{c,1})$ attains arbitrarily large values: Note that for $r_{c}$ as in \eqref{eq:modified} we have
\[r_{c}\leq \|{g}\|_\infty+\log(1-\delta_c)\lim_{T\rightarrow\infty}\frac{1}{T}\E^{K_c}_\pi(|\{n:\tau_{n,c}\leq T\}|)=\|{g}\|_\infty+\log(1-\delta_c)\frac{1}{\E^{K_c}_{\nu_c}(\tau_{c,1})},\]
where we used Proposition \ref{prop:renewal} (iii) in the last step.
Since $r_c\in \R_{+}$ and $\log(1-\delta_c)=-c\rightarrow-\infty$ for $c\rightarrow\infty$, we must have $\E^{K_c}_{\nu_c}(\tau_{c,1})\rightarrow \infty$ for $c\rightarrow\infty$.\\
We furthermore show that $\E^{K_c}_{\nu_c}(\tau_{c,1})$ attains arbitrarily small values: By \cite{CIL}, there exist $A,B\in(0,1)$ such that $a_{c_n},\alpha_{c_n}\rightarrow A$ and $b_{c_n},\beta_{c_n}\rightarrow B$ for some null sequence $(c_n)_{n\in\N}$. (Note that the assumption $\gamma>0$ in the above reference is indeed not relevant for this fact.)
Using the representation in Proposition \ref{prop:renewal} \eqref{renewal:(iv)}, we obtain that this implies
\[\E^{K_{c_n}}_{\nu_{c_n}}(\tau_{c,1})\leq \max\{|h_1(a_{c_n})-h_1(\alpha_{c_n})|,|h_1(b_{c_n})-h_1(\beta_{c_n})|\}\rightarrow 0\;\;\mbox{ as $n\rightarrow\infty$}\]
since $h_1$ is a continuous function {on $(0,1)$}.\\
By the implicit functions theorem applied to the system of equations given in Lemma \ref{lem:fixed_costs}, we furthermore obtain that $c\mapsto(a_c, \alpha_c,\beta_c, b_c)$ is a continuous function, and - using Proposition \ref{prop:renewal} again - we obtain that $c\mapsto \E^{K_c}_{\nu_c}(\tau_{c,1})$ is also continuous. \\
Putting pieces together, the intermediate value theorem yields the desired result.
\end{proof}

Now, we obtain the existence result:
\begin{proposition}\label{prop:existence}
Let $h>0$. Then there exist $v_h:(0,1)\rightarrow\R$, $\lambda_h\in\R$, and $0<a_h< \alpha_h\leq\beta_h< b_h<1$ such that assumptions $(i)$-$(vii)$ in Theorem \ref{thm:verification} hold.
\end{proposition}
\begin{proof}
By Lemma \ref{lem:h_c}, there exists {$c>0$} such that {$h=\E^{K_c}_{\nu_c}(\tau_{c,1})$}, where we have chosen the corresponding parameters according to Lemma \ref{lem:fixed_costs}, that satisfy $(i)$-$(vi)$ in Theorem \ref{thm:verification}. By Proposition \ref{prop:renewal} \eqref{renewal:(iii)} we see that {condition} (vii) of Theorem \ref{thm:verification} also holds true.
\end{proof}

\begin{theorem}
For each $h>0$ there exist parameters $a,\alpha,\beta,b\in(0,1)$ and $c>0,\lambda\in\R$ such that the constant boundary strategy with parameters $(a,\alpha,\beta,b)$ is optimal for \eqref{eq:value}, and
\[V^{o}(h)=\lambda+\frac{c}{h}.\]
\end{theorem}
\begin{proof}
The result immediately follows from Proposition \ref{prop:existence} together with Theorem \ref{thm:verification}.
\end{proof}

\begin{remark}\label{rem:construction}
Note that our line of argument is constructive, since we have characterized the parameters as a solution of a system of equations, that can be solved using standard numerical methods. A second way for finding the optimal parameters is to make use of Proposition \ref{prop:renewal}: Since all expressions are explicitly given, we may optimize \[r_K=\frac{\E^{K}_\nu g(\pi_{{\tau_1}},\phi(\pi_{{\tau_1}}))}{\E^{K}_\nu (\tau_1)},\]
over all parameters $a<\alpha\leq \beta<b$ under the restriction that
\[{\E^{K}_\nu{(\tau_1)}}=h,\]
{where all notations are given according to Subsection \ref{subs:constant_boundary}.} This can also be carried out using standard numerical procedures and seems to be more stable than the first way in our numerical examples.
\end{remark}
{
\begin{remark}
It is remarkable to note that our portfolio optimization problem with pure proportional transaction costs and a restricted set of trading strategies leads to similar equations as the problem with fixed and proportional transaction costs as given in  \cite{T06,T08}, but with the additional parameter $c$. In that reference the cost structure arises due to the assumption that the fixed and proportional transaction costs have to be paid from both the bond- and the stock-account, see Theorem \ref{thm:verification} and Remark \ref{anmerk:tamura}.
\end{remark}
}

\subsection{Explicit example}\label{subsec:explicit}
In this {sub}section, we want to give an explicit example for the previous results. We first consider the case of vanishing proportional transaction costs, that is $\gamma=0$, and a Merton ration $\pi^{*}=\tfrac{1}{2}$. In this framework, we are able to compute an explicit expression of the growth rate depending on the financial market and the average frequency of trading. Due to the symmetry of the situation it is natural to make the ansatz to choose the stopping boundaries $a,b$ symmetric around $\pi^*$, i.e. $1-a=b$, and to choose $\alpha=\beta=\pi^*$. Using Proposition \ref{prop:renewal} we have
\[\E^{K}_{\pi^*}(\tau)=h_1(a),\]
where $\tau=\inf\{t\ge 0:\pi_t\not\in[a,b]\}$ and
\[h_1(x)=\frac{1}{\sigma^2}\left(\log\frac{1-x}{x}\right)^{2}.\]
Hence
$\lim_{T\rightarrow\infty} \frac{1}{T}\E^{K}_{\pi^{*}}(|\{n:\tau_n\leq T\}|)=\frac{1}{h}$ if and only if
\[b=b_h=\frac{e^{\sigma\sqrt{h}}}{1+e^{\sigma\sqrt{h}}},\;\;\;a=a_h=1-b_h=\frac{1}{1+e^{\sigma\sqrt{h}}}.\]
It remains to be checked that the constant boundary strategy with parameters $(a_h,\pi^*,\pi^*,b_h)$ is indeed optimal. This can be verified immediately by applying Theorem \ref{thm:verification} with
\[\lambda=\frac{\sigma^2(b-\frac{1}{2})}{2\log(\frac{b}{1-b})},\;\;\;v(x)=\int_{\pi^*}^xg(y)dy,\;\;\;c=v(\pi^*)-v(b),\]
where
\[g(y)=\frac{1}{y(1-y)}\left(\frac{2\lambda}{\sigma^2}\log\frac{y}{1-y}+\frac{1}{2}-y\right).\]
Finally the growth rate is given by
\[V^{o}(h)=\frac{\log(1/2)-\log\left(\frac{e^{\sigma\sqrt{h}/2}}{1+e^{\sigma\sqrt{h}}}\right)}{h}.\]
A short calculation furthermore yields that $V^{o}(h)\rightarrow V^M$ for $h\rightarrow 0$, as expected. More precisely,
\[V^{o}(h)=\frac{\sigma^2}{8}-\frac{\sigma^4}{192}h+O(h^2).\]

\section{Numerical examples}\label{sec:example}

\subsection{Review of existing time-discretizations of the Merton problem without transaction costs}\label{sec:discrete}
To be optimal, the {\bf Merton-investor} with transaction costs {keeps his risky fraction process constant.} But since the stock price follows a geometric Brownian motion his wealth is changing at all times and so he has to adjust his portfolio continuously. Therefore, this investor may be considered as an unrealistic one. One way out is to restrict the trading times to those that the investor adjusts his portfolio only at times which are multiples of a fixed $h>0$. This leads to more realistic trading strategies. This setting was considered in \cite{relaxed}, and we follow this treatment in this subsection. In the following we call such an investor an {{\textbf{$\boldsymbol h$-investor}}.}
By using the theory of Markovian Decision Processes (see \cite{BR}) the following theorem can be easily shown.

\begin{proposition} Let $T=N\cdot h$ for $h>0$ and $N\in\mathds{N}$. Then
$$ \sup_{\pi} \mathds{E}\big{[} \log(X^{\pi}_{T}) \big{]}= N \cdot A(h) + \log(x)\;,$$
where
$$ A(h) :=  \sup_{a \in [0,1]} \E \big{[}  \log(aZ+(1-a))\big{]}\;, \quad  Z=Z(h) :=\exp(\sigma W_h + (\mu- \sigma^2 /2)h). $$
\label{001} Here, the supremum is taken over all strategies $\pi$, that allow to rebalance the portfolio at all discrete time points which are multiples of $h>0$.
\end{proposition}

\begin{remark}
$ Z $ is the random return of the stock in the time period $h$.
\end{remark}

Using Proposition \ref{001}, we get
\begin{align*}
\sup_{\pi}\liminf_{T\rightarrow \infty} \frac{1}{T}\E\big{(} \log(X^{\pi}_T) \big{)} \le \liminf_{T\rightarrow \infty} \frac{1}{T} \sup_{\pi} \E\big{(} \log(X^{\pi}_T) \big{)}   = \frac{A(h)}{h}\;.
\label{541}
\end{align*}

Let $a^{*}$ be the number which maximizes $ a \rightarrow \E \big{[}  \log(aZ+(1-a))\big{]}$ and let $\pi^{*}$ be the trading strategy which invests this fraction $a^{*}$ in the stock at all discrete time points which are multiples of $h$. Hence the growth rate with respect to $\pi^{*}$ equals $\tfrac{A(h)}{h}$. This yield the following proposition.

\begin{proposition} The growth rate $V^{\mbox{h-invest}}(h)$ of the h-investor equals $ \frac{A(h)}{h}$.
\end{proposition}

Since it is not possible to find $A(h)$ in closed form, the following approximation is useful: Similar calculations as in \cite[Section 3]{relaxed} based on a Taylor expansion yield that for large $N \in \mathds{N}$ and small $h$
\begin{equation*}
A(h) \approx \sup_{a\in[0,1]} \sum_{k=1}^{N} (-1)^{k+1} \frac{\E(Y^k(h))}{k},
\end{equation*}
where $Y(h) = a(Z(h)-1)$.

A further idea of limiting the trading times is the widely used approach to restrict the trading times to those at which an independent Poisson process jumps. Such situations are for example studied in \cite{MatLog, MatGrow} and \cite{PTCM}. In the following we call such an investor an {\textbf{$\boldsymbol \lambda$-investor}}.

Let $V^{\mbox{$\lambda$-invest}}(\lambda)$ be the optimal growth rate of the $\lambda$-investor. Since this investor trades only at the jump times of an independent Poisson process with intensity $\lambda >0$, the average number of trading times per unit is given by ${\lambda}$. From \cite{MatGrow}, Theorem 5.2, we get the following limit theorem for the value $V^{\mbox{$\lambda$-invest}}(\lambda)$.

\begin{proposition} $\quad \lambda(V^{M}-V^{\mbox{$\lambda$-invest}}(\lambda)) \rightarrow \frac{1}{2}\mu^2(1-\frac{\mu}{\sigma^2})^2\quad $ as $\; \lambda\rightarrow \infty$.
\end{proposition}
This yield directly the following approximation: The growth rate $V^{\mbox{$\lambda$-invest}}(\lambda)$ can be approximated for large $\lambda$ by $$V^{M}-\frac{1}{2}\mu^2(1-\frac{\mu}{\sigma^2})^2 \frac{1}{\lambda}\,.$$

\subsection{Efficiency}
In this section, we want to compare the different investors. As proved above, we know that the best possibility for an investor, who does not want to adjust his portfolio more than $1/h$ times per time unit {in average}, is to use a constant boundary strategy. Now we want to compare this performance with the performance of the $h$-investor and the $\lambda$-investor introduced in Section \ref{sec:discrete}.

To compare the different type of investors, we need a performance criterion. Similar to \cite{relaxed}, we use the natural measure of the efficiency.

\begin{definition} The efficiencies of the $h$-investor, the $\lambda$-investor and the $o$-investor without transaction costs ($\gamma=0$) are given for $h>0$ by
\begin{align*}
E^{h}(h) = \frac{V^{\mbox{h-invest}}(h)}{V^{M}} \;, \qquad E^{\lambda}(h) = \frac{V^{\mbox{$\lambda$-invest}}({\frac{1}{h}})}{V^{M}} \qquad \text{ and } \qquad  E^{o}(h) = \frac{V^{o}(h)}{V^{M}}\;,
\end{align*}
respectively. With transaction costs the efficiency of the $o$-investor is given by
\begin{align*}
E^{o}_c(h) = \frac{V^{o}(h)}{V^{M}_c}\;.
\end{align*}
\end{definition}

\subsection{Comparison of the efficiencies without transaction costs}

Here in this subsection, we want to illustrate the results of the previous sections. The efficiencies of the $h$- and $\lambda$-investor can easily be obtained by using the results of Section \ref{sec:discrete}. If the Merton ratio $\tfrac{\mu}{\sigma^2}$ equals $\tfrac{1}{2}$, then we have already computed explicitly the growth rate of the $o$-investor depending on $h$, {cf. Subsection \ref{subsec:explicit}.} In the general case, the efficiency of the $o$-investor can be computed for a fixed $h>0$ by using Lemma \ref{lem:fixed_costs}. Then interpolating the efficiencies for different $h$ on a dense grid yield the efficiency curve of this investor.

\subparagraph{Example 1:} As basis for this numerical example, we fix the coefficients of the Black-Scholes market to {$\mu=8\%$ and $\sigma=40\%$}. Hence the Merton ratio $\tfrac{\mu}{\sigma^2}$ equals $\tfrac{1}{2}$. The results are shown in Figure \ref{fig1}. As we can see, the $o$-investor is doing much better than the $h$- and $\lambda$-investor. This can be explained in the following way: The $h$- and $\lambda$-investors adjust their portfolio independently of the evolvement of the financial market while the $o$-investor trades when the risky fraction process is too much unbalanced. Hence the $o$-investor saves trades in calm times that he can spend them in strenuous ones to improve his performance.

\begin{figure}[t]
\center
\includegraphics[width=8cm]{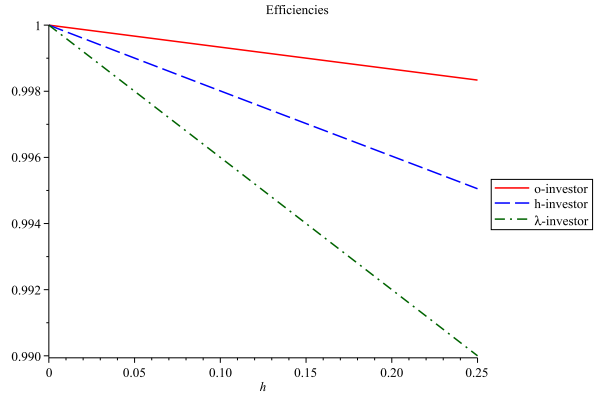}
\caption{Efficiencies of the $o$-, $h$- and $\lambda$-investor when the Merton ratio equals $\tfrac{1}{2}$.}
\label{fig1}
\end{figure}

\subparagraph{Example 2:} As basis for this numerical example, we fix the coefficients of the Black-Scholes market to {$\mu=\tfrac{8}{100}$ and $\sigma=\sqrt{\tfrac{2}{15}}$}. Hence the Merton ratio  $\tfrac{\mu}{\sigma^2}$ equals $\tfrac{3}{5}$. The results are shown in Figure \ref{fig2}. By the same arguments as in Example 1, the $o$-investor is performing much better.

\begin{figure}[t]
\center
\includegraphics[width=8cm]{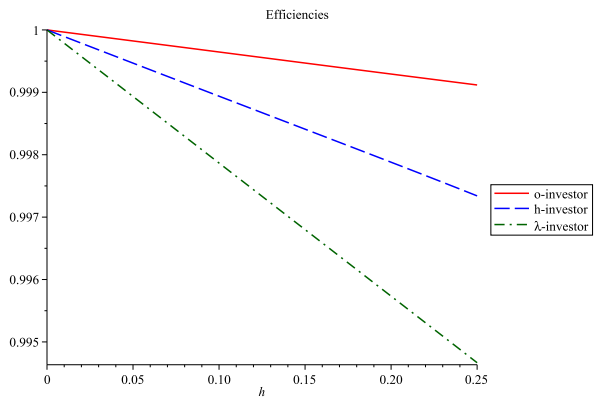}
\caption{Efficiencies of the $o$-, $h$- and $\lambda$-investor when the Merton ratio equals $\tfrac{3}{5}$.}
\label{fig2}
\end{figure}

\subsection{Efficiency with transaction costs}
The previous examples without transaction costs already show that even quite large values of $h$ only lead to a small loss of efficiency. This means that -- when using the optimal constant-boundary strategy -- restricting oneself to not trading too often in the long run does not influence the asymptotic growth rate too much. For a model with (proportional) transaction costs, the same result can be observed. As an example, we consider the model with parameters $\mu=9.6\%$, $\sigma=40\%$ and transaction costs of $\gamma=0.3\%$. We found the results using the second approach described in Remark \ref{rem:construction}. The results are given in Figure \ref{fig3} and \ref{fig4}. The effect is even stronger than in the examples without transaction costs: When using $h=0.2$, the efficiency is still larger than $99.995\%$: The asymptotic growth rate shrinks from $2.84795\%$ in the unrestricted setting (see Section 2.2) to $2.84782\%$ in our restricted setting with $h=0.2$. So one can argue that in the market with proportional transaction costs a \glqq relaxed\grqq\; investor who uses a constant boundary strategy has nearly no drawback compared to to an investor who follows the optimal strategy.

\begin{figure}[t]
\center
\includegraphics[width=7cm]{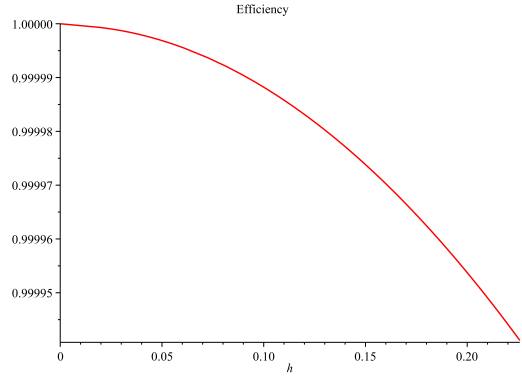}
\caption{Efficiency of the $o$-investor with transaction costs $\gamma=0.3\%$.}
\label{fig3}
\end{figure}

\begin{figure}[t]
\center
\includegraphics[width=7cm]{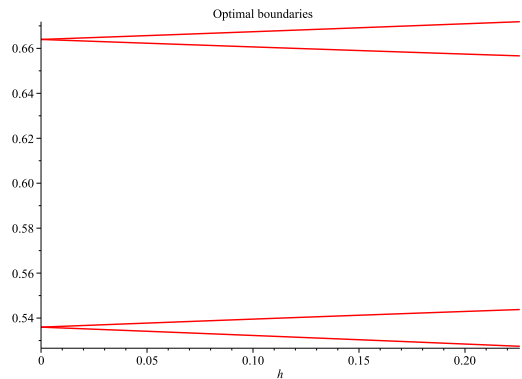}
\caption{Optimal boundaries of $o$-investor with transaction costs $\gamma=0.3\%$.}
\label{fig4}
\end{figure}

This discussion also sheds light on good portfolio strategies in cases with constant transaction costs as discussed, e.g., in \cite{K98}. Here, the investor has to pay a constant amount $c>0$ for each transaction. In this case, the optimal strategy does not only depend on the fraction invested into the stock, but also on the total wealth. This makes the more complicated and constant boundaries as described before cannot expected to be optimal anymore. But the numerical results given before suggest constant-boundary strategies with long average waiting times between each two transactions perform very good. Using these strategies, the constant costs only have to be paid each $h$ time units (years). Therefore, even for a moderate portfolio value, the results suggest that constant boundary strategies seem to be good strategies also in this case.

\bibliographystyle{apalike2}
\bibliography{bibliography}

\end{document}